\documentclass[UKenglish]{lipics-v2016}

\usepackage[utf8]{inputenc}
\usepackage{amsmath}
\usepackage{amsfonts}
\usepackage{amssymb}

\usepackage{tikz}
\usepackage{pgfplots}
\usetikzlibrary{patterns,fit, calc, positioning}
\pgfplotsset{width=7cm, compat=1.10}
\usepgfplotslibrary{fillbetween}

\usepackage{subcaption}

\newcommand{\bin}{{\{0,1\}}}
\newcommand{\Note}[1]{{\color{red}#1}}
\newcommand{\eqdef}{\stackrel{{\sf def}}{=}}

\DeclareMathOperator{\E}{\mathbb{E}}
\newcommand{\D}{\mathsf{D}}

\renewcommand{\H}{{\sf H}}
\newcommand{\Hmin}{\H_\infty}
\newcommand{\Hmins}[1]{\Hmin^{#1}}
\newcommand{\Hhill}[1]{\H^{\sf HILL}_{#1}}

\newtheorem{corollary}{Corollary}

\newtheorem{lemma}{Lemma}

\usepackage[backend=bibtex,style=numeric-verb,maxnames=5,minnames=5,maxbibnames=99,firstinits=true,
doi=true,isbn=true,url=true
]{biblatex}
\AtEveryBibitem{%
    \clearfield{pages}
}

\usepackage{hyperref}
\hypersetup{
  colorlinks=true,
  pdfmenubar=false,
  pdfstartview={FitH},
  linktoc=all,
  urlcolor=blue,
  linkcolor={blue},
}
\usepackage[nameinlink]{cleveref}
\crefname{table}{table}{tables}
\Crefname{table}{Table}{Tables}
\crefname{figure}{figure}{figures}
\Crefname{figure}{Figure}{Figures}
\Crefname{tikzpicture}{Figure}{Figures}
\crefname{section}{section}{sections}
\Crefname{section}{Section}{Sections}
\crefname{claim}{claim}{claims}
\Crefname{claim}{Claim}{Claims}
\Crefname{algorithm}{Algorithm}{Algorithms}
\crefname{lemma}{Lemma}{Lemmas}
\Crefname{lemma}{Lemma}{Lemmas}
\Crefname{corollary}{Corollary}{Corollaries}
\crefname{corollary}{Corollary}{Corollaries}
\Crefname{theorem}{Theorem}{Theorems}
\Crefname{definition}{Definition}{Definitions}

\author[1]{Krzysztof Pietrzak\thanks{Supported by the European Research Council, ERC consolidator grant (682815 - TOCNeT).}}
\affil[1]{IST Austria\\
  \href{mailto:pietrzak@ist.ac.at}{pietrzak@ist.ac.at}}
\author[2]{Maciej Skorski\thanks{Supported by the European Research Council, ERC consolidator grant (682815 - TOCNeT).}}
\affil[2]{IST Austria\\
  \href{mailto:mskorski@ist.ac.at}{mskorski@ist.ac.at}}
\authorrunning{K.\,Pietrzak and M.\,Skorski} 

\Copyright{Krzysztof Pietrzak and Maciej Skorski}


\EventEditors{Ioannis Chatzigiannakis, Piotr Indyk, Fabian Kuhn, and Anca Muscholl}
\EventNoEds{4}
\EventLongTitle{44th International Colloquium on Automata, Languages, and Programming (ICALP 2017)}
\EventShortTitle{ICALP 2017}
\EventAcronym{ICALP}
\EventYear{2017}
\EventDate{July 10--14, 2017}
\EventLocation{Warsaw, Poland}
\EventLogo{eatcs}
\SeriesVolume{80}
\ArticleNo{}

\title{Non-Uniform Attacks Against Pseudoentropy\footnote{The full version is available at \url{https://arxiv.org/abs/1704.08678}}}

\begin{document}
\maketitle


\begin{abstract}
De, Trevisan and Tulsiani [CRYPTO 2010] show that every distribution over $n$-bit strings which has constant statistical distance to uniform (e.g., the output of a pseudorandom generator mapping $n-1$ to $n$ bit strings), can be distinguished from the uniform 
distribution with advantage $\epsilon$ by a circuit of size $O( 2^n\epsilon^2)$.

We generalize this result, showing that a distribution which has less than $k$ bits of min-entropy, can be distinguished from any distribution with $k$ bits of $\delta$-smooth min-entropy with advantage 
$\epsilon$ by a circuit of size $O(2^k\epsilon^2/\delta^2)$. 
As a special case, this implies that any distribution with support at most $2^k$ (e.g., the output of a pseudoentropy generator 
mapping $k$ to $n$ bit strings) can be distinguished from any 
given distribution with min-entropy $k+1$ with advantage $\epsilon$ by a circuit of size $O(2^k\epsilon^2)$.

Our result thus shows that pseudoentropy distributions face basically the same non-uniform attacks 
as pseudorandom distributions.

\end{abstract}

\section{Introduction}

De, Trevisan and Tulsiani~\cite{DeTT10} show a non-uniform attack against any pseudorandom generator 
(PRG) which maps $\bin^{n-1}\rightarrow \bin^n$. For any $\epsilon\ge 2^{-n/2}$, their attack achieves distinguishing advantage $\epsilon$ and can be realized by a circuit of size  $O\left( 2^{n}\epsilon^2\right)$. Their attack doesn't even
need the PRG to be efficiently computable.

In this work we consider a more general question, where we ask for attacks distinguishing a distribution from any distribution with slightly higher min-entropy. We generalize~\cite{DeTT10}, showing a non-uniform attack which, for any $\epsilon,\delta>0$, distinguishes any distribution with $<k$ bits of min-entropy from any distribution with $k$ bits of $\delta$-smooth min-entropy 
with advantage $\epsilon$, and where the distinguisher is of size  $O(2^{k}\epsilon^2/\delta^2)$.  
As a corollary we recover the~\cite{DeTT10} result, showing that the output of any pseudoentropy generator $\bin^{k}\rightarrow \bin^n$ can be distinguished from 
any variable with min-entropy $k+1$ with advantage $\epsilon$ by circuits of size $O(2^k\epsilon^2)$.
\begin{itemize}
 \item From a theoretical perspective, we prove where the separation between pseudoentropy and smooth min-entropy lies, by classifying how powerful computationally bounded adversaries can be so they can still be fooled to ``see'' more entropy than there really is.
 \item From a more practical perspective, our result shows that using pseudoentropy instead of pseudorandomness (which 
 for many applications is sufficient and allows for saving in entropy \emph{quantity}~\cite{DodisPietrzakWichs2013}), will not give improvements in terms of \emph{quality} (i.e., the size and advantage of distinguishers considered), at least not against generic non-uniform attacks.
\end{itemize}

\subsection{Notation and Basic Definitions}
Two variables $X$ and $Y$ are $(s,\epsilon)$ indistinguishable, denoted $X\sim_{s,\epsilon}Y$, if for all boolean circuits $D$ 
of size $|D|\le s$ we have $|\Pr[D(X)=1]-\Pr[D(Y)=1]|\le \epsilon$. The statistical distance 
of $X$ and $Y$ is $d_1(X;Y)\eqdef\sum_x|P_X(x)-P_Y(x)|$ (where $P_X(x)\eqdef \Pr[X=x]$), 
the Euclidean distance of $X$ and $Y$ is $d_2(P_X;P_Y) \eqdef \sqrt{\sum_{x}(P_X(x)-P_Y(x))^2}$.
A variable $X$ has min-entropy $k$ if it doesn't take any particular outcome 
with probability greater $2^{-k}$, it has $\delta$-smooth min-entropy $k$~\cite{DBLP:conf/asiacrypt/RennerW05}, if it's $\delta$ close to some distribution with min-entropy $k$. $X$ has $k$ bits of HILL pseudoentoentry of quality $(s,\epsilon)$ if there exists a $Y$ with min-entropy $k$ that is $(s,\epsilon)$ indistinguishable from $X$, we use the following standard notation for these notions
\begin{description}
\item{min-entropy: }$\Hmin(X)\eqdef -\log\max_x\left( \Pr[X=x]\right)\ .$
\item{smooth min-entropy: }$\Hmins{\delta}(X)\eqdef \max_{Y,d_1(X;Y)\le \delta}\Hmin(Y)\ .$
\item{HILL pseudoentropy: }$\Hhill{s,\epsilon}(X)\eqdef \max_{Y,Y\sim_{(s,\epsilon)}X}\Hmin(Y)\ .$
\end{description}

\subsection{Our Contribution}
In this work give generic non-uniform attacks on pseudoentropy distributions. 
A seemingly natural goal is to consider a distribution $X$ with $\Hmin(X)\le k$ bits of min-entropy, strictly larger $\Hhill{s,\epsilon}(X)\ge k+1$ bits of  HILL entropy, and then give an upper bound on $s$ in terms of $\epsilon$. 
This does not work as there are $X$ where $\Hmin(X)\ll \Hmins{\delta}(X)$,\footnote{Consider 
an $X$ which is basically uniform over $\bin^n$, but has mass $\delta$ on one particular point, then 
$\log\delta^{-1}=\Hmin(X)\ll \Hmins{\delta}(X)=n$.} and as by definition $\Hmins{\delta}(X)=\Hhill{\infty,\delta}(X)$, 
we can have a large entropy gap $\Hhill{\infty,\delta}(X)-\Hmin(X)$ even when considering unbounded adversaries against 
HILL entropy. For this reason, in our main technical result~\Cref{thm:main} below, we must consider distributions 
with bounded \emph{smooth} min-entropy. This makes the statement of the lemma somewhat technical. 
In practice, the distributions considered often have bounded support, for example because they were generated from a short seed by a deterministic process (like a pseudorandom generator). In this case we can drop the smoothness requirement as stated in \Cref{cor:main} below.

\begin{lemma}[Nonuniform attacks against pseudoentropy]\label{thm:main}
Suppose that $X\in\bin^n$ does not have $
k$ bits of $\delta$-smooth min-entropy, i.e., $\Hmins{\delta}(X)<k$, then for any $\epsilon$ we have 
\begin{align*}
\Hhill{\tilde{O}( 2^k \epsilon^2\delta^{-2}),\epsilon}(X) < k
\end{align*}
where $\tilde{O}(\cdot)$ hides a factor linear in $n$.
\end{lemma}
\begin{theorem}\label{cor:main}
Let $f:\bin^k\rightarrow\bin^n$ be a deterministic (not necessarily efficient) function. Then we have 
\begin{align*}
 \Hhill{\tilde{O}( 2^k \epsilon^2),\epsilon}(f(U_k)) \le k+1.
\end{align*}
more generally, for any $X$ over $\bin^n$ with support of size $\le 2^k$
\begin{align*}
 \Hhill{\tilde{O}( 2^k \epsilon^2),\epsilon}(X) \le k+1.
\end{align*}
\end{theorem}
\begin{remark}[Concluding best attacks against PRGs]
For the special case $n=k+1$ we recover the bound for 
pseudo\emph{random} generators from~\cite{DeTT10}.
\end{remark}
\begin{proof}[Proof of~\Cref{cor:main}]
The theorem follows from~\Cref{thm:main} when $\delta=1/2$; consider any $X$ with support of size $\le 2^k$, then 
$\Hmins{\delta}(X)\le k+1$, as no matter how we cut probability mass of $1-\delta=1/2$ over $
2^k$ elements, one element will have the weight at least $2^{-k-1}$.
\end{proof}
\subsection{Proof Outline}

\subsubsection{A Weaker Result as a Ball-Bins Problem}

We outline the proof of a somewhat weakened version of \Cref{cor:main} in the language of balls and bins. For every $Y$ of min-entropy $k'=k+\Omega(1)$ we want to distinguish
$Y$ from $X = f(U_k)$. Suppose for simplicity that $Y$ is flat and $f$ is injective, so that $X$ is also flat.
Our strategy will be to hash the points randomly into two bins and take advantage of the fact that the \emph{average maximum load}
is closer to $\frac{1}{2}$ when we sample from $Y$ than when drawing from $X$. The reason is that $Y$ has more balls, so by the law of large numbers, we expect the load to be ``more concentrated'' around the mean.

Think of throwing balls (inputs $x$) into two bins (labeled by $-1$ and $1$).
If the balls come from the support of $X$, the expected maximum load (over two bins) 
equals $\approx 2^{k-1} + \sqrt{2/\pi}\cdot 2^{k/2}$. Similarly, 
if the balls come from the support of $Y$, then maximum load is $2^{k'-1} + \sqrt{2/\pi}\cdot 2^{k'/2}$.
In terms of the average load (the load normalized by the total number of balls)
\begin{align*}
 \mathsf{AverageMaxLoad}(X) \approx 0.5 + \sqrt{2/\pi}\cdot 2^{-k/2} \quad \text{ w.h.p. when drawing from } X \\
 \mathsf{AverageMaxLoad}(Y) \approx 0.5 + \sqrt{2/\pi}\cdot 2^{-k'/2} \quad \text{ w.h.p. when drawing from } Y 
\end{align*}
As $k'=k+\Omega(1)$ we obtain (with good probability)
\begin{align*}
  \mathsf{AverageMaxLoad}(X) -  \mathsf{AverageMaxLoad}(Y) =  \Omega(2^{-k/2}).
\end{align*}
Letting $\D$ be one of these bins assignments we obtain a distinguisher with advantage $\epsilon = \Omega(2^{-k/2})$.
To generate the assignments efficiently we relax the assumption about choosing bins and assume 
only that the choices of bins are independent for any group of $\ell=4$ balls. The fourth moment method allows us to keep sufficiently good
probabilistic guarantees on the maximum load.

\subsubsection{The General Case by Random Walk Techniques}

\paragraph{A high-level outline and comparison to \cite{DeTT10}}
Below in \Cref{fig:outline} we sketch the flow of our argument. 

\begin{figure}[h!]
\resizebox{0.99\textwidth}{!}{
\begin{tikzpicture}
 \node[draw, rectangle] (small_entr) at (0,0) {$X$ has no smooth-min entropy $k$};
 \node[draw, rectangle, align=center] (mass_conc) at (0,-2) {large bias between $X$ and $Y$ on \emph{only} $2^k$ elements }; 
 \node[draw, rectangle, align=center] (big_dist) at (0,-4) {$d_2(X;Y)=\Omega(2^{-\frac{k}{2}})$ (Euclidean distance)}; 
\node[draw, rectangle, align=center] (random_attack) at (13.0,-4) {advantage of random attack \\ $\epsilon \approx d_2(X;Y) $ \\
for any $X,Y$ (\Cref{lemma:random_attack_advantage})}; 
\node[draw, rectangle, align=center] (low_complex_attack) at (5,-6) {a random distinguisheer $\D$ achieves $\epsilon = \Omega\left( 2^{-\frac{k}{2}} \right) $ }; 
\node[draw, rectangle, align=center] (slice) at (5,-8) {$\epsilon = \Omega\left( T^{-\frac{1}{2}} 2^{-\frac{k}{2}} \right)$ for a random $\D$ restricted to one slice };
\node[draw, rectangle, align=center] (slices_together) at (5,-10) {$\epsilon = T\cdot \Omega\left( T^{-\frac{1}{2}} 2^{-\frac{k}{2}} \right)$ by composing
advantages from all slices \\ (needs $O(T)$ advice)
}; 
\node[draw, rectangle, align=center] (continous_threshold) at (5,-12) {arbitrary $\epsilon$ in size $2^{\frac{k}{2}}\epsilon$ (by manipulating $T$)};
\node[draw, rectangle, align=center] (relaxing_randomness) at (5,-14) {weak randomness for distinguishers and slices is enough \\ ($4$-wise independence works!)};

\node[rectangle, draw, dotted, inner sep = 3mm, fit=(slice) (slices_together)] (slices_comment) {};
\node at (slices_comment.north west) {domain partitioned \emph{randomly} into $T$ slices};
\node[rectangle, draw, dotted, inner sep = 3mm, fit=(mass_conc) (relaxing_randomness) (slices_together) (slice)] (Y_comment) {};
\node at (Y_comment.north east) {for any \emph{fixed} $Y$ of min-entropy at least $k$};
\draw[->] (continous_threshold)  -- (relaxing_randomness) node[midway, right, align=center]{random walks moment inequalities \\ (see \Cref{sec:random_walks,sec:interpol,sec:anticon_bounds})}; 
\draw[->]  (slices_together) -- (continous_threshold) node[midway, right] {\Cref{cor:tradeoof}}; 
\draw[->] (slice) -- (slices_together) node[midway, right]{\Cref{cor:parition_advantage,cor:parition_advantage_final}};
\draw[->] (low_complex_attack) -- (slice) node[midway,right]{\Cref{cor:slice_advantages}};
\draw[->] (small_entr) -- (mass_conc) node[midway, right]{\Cref{lemma:smooth_charact,cor:no_smoothentropy_implies_bias_to_entropy}};
\draw[->] (mass_conc) -- (big_dist) node[midway, right]{\Cref{lemma:bias_to_euclidean,cor:no_smooth_implies_euclidean_distance}
};
\draw[->] (big_dist) -| (low_complex_attack) node[above right = 0.5cm and -3.7cm of low_complex_attack]{\Cref{cor:random_attack_advantage}};
\draw[->] (random_attack) -| (low_complex_attack);
\end{tikzpicture}
}
\caption{The map of our proof.}
\label{fig:outline}
\end{figure}
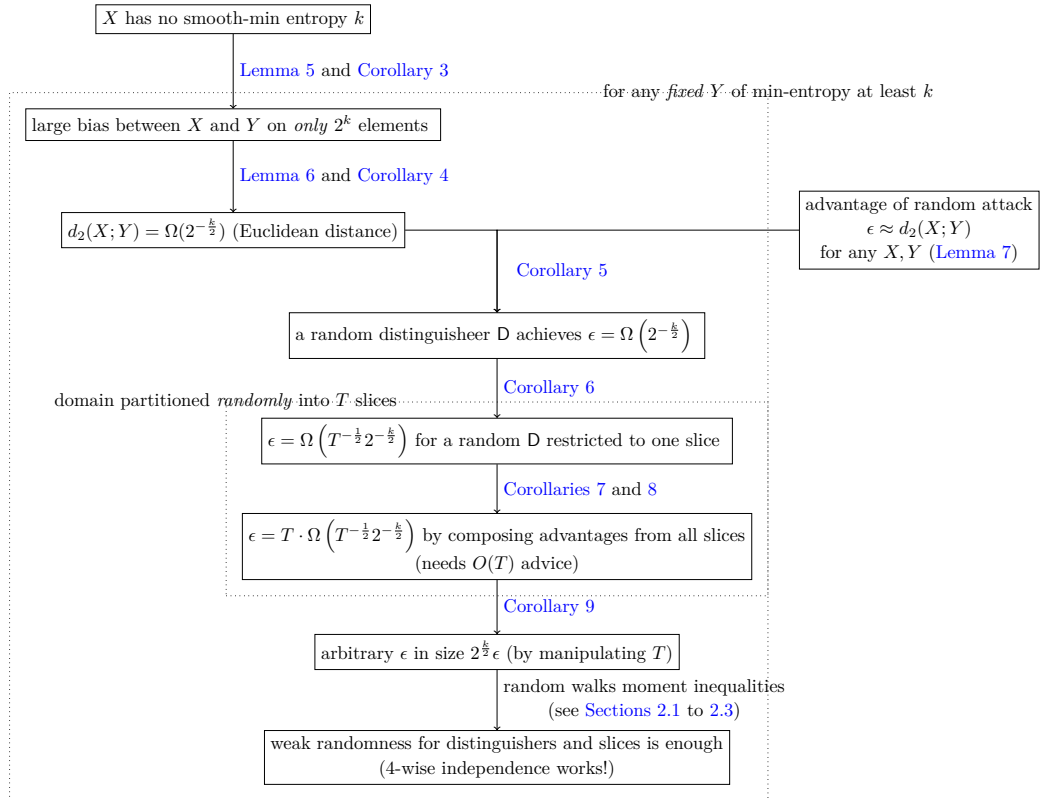
Our starting point is the proof from~\cite{DeTT10}. They use the fact that a random mapping $\D:\bin^n\rightarrow \{-1,1\}$ likely distinguishes
any two distributions $X$ and $Y$ over $\bin^n$ with advantage being the Euclidean distance $d_2(X;Y)\eqdef \sqrt{ \sum_{x}(P_X(x)-P_Y(x))^2 }$.


For any $X$ and $Y$ with constant statistical distance $\sum_{x}|P_X(x)-P_Y(x)|=\Theta(1)$ (which is the case for the PRG setting where $Y=U_n$ and 
$X=PRG(U_{n-1})$) 
this yields a bound $\Omega\left(2^{-\frac{n}{2}}\right)$. This bound can be then amplified, at the cost of extra advice, by partitioning the domain $\bin^n$ and combining
corresponding advantages (advice basically encodes if there is a need for flipping the output). Finally one can show that
$4$-wise independence provides enough randomness for this argument, which makes sampling $\D$ efficient. Our argument deviates from this approach in two important aspects. 

The first difference is that in the pseudoentropy case
we can improve the advantage from $\Omega\left(2^{-\frac{n}{2}}\right)$, where $n$ is the logarithm of the support of the variables considered, to $\Omega\left(2^{-\frac{k}{2}}\right)$, where $k$ is the min-entropy of the variable we want to distinguish from. The reason is that being statistically far from any $k$-bit min-entropy distributions implies a \emph{large bias on already $2^k$ elements}.
This fact (see \Cref{lemma:smooth_charact,cor:no_smoothentropy_implies_bias_to_entropy}, and also \Cref{fig:attack}) is a new characterization of smooth min-entropy of independent interest.

The second subtlety arises when it comes to amplify the advantage over the partition slices. For the pseudorandomness case
it is enough to split the domain in a deterministic way, for example by fixing prefixes of $n$-bit strings, 
in our case this is not sufficient. 
For us a ``good'' partition must shatter the $2^k$-element high-biased set, which can be arbitrary. 
Our solution is to use \emph{random partitions}, in fact, we show that using $4$-universal hashing is sufficient. 
Generating base distinguishers and partitions at the same time makes probability calculations more involved.

Technical calculations are based on the fourth moment method, similarly as in \cite{DeTT10}. 
The basic idea is that for settings where the second and fourth moment are easy to compute (e.g. sums of independent symmetric random variables)
we can obtain good {upper and lower} bounds on the first moment. In the context of algorithmic applications these techniques are usually credited to~\cite{berger1991fourth}.
Interestingly, exploiting natural relations to \emph{random walks}, we show that calculations immediately follow by adopting classical (almost one century old) tools and results~\cite{marcinkiewicz1937quelques,khintchine1924satz}. 
Our technical novelty is an application of moment inequalities
due to Marcinkiewicz-Zygmund and Paley-Zygmund, which allow us to prove slightly more than just the existence of an attack. Namely
we generate it with {constant success probability}.




\paragraph{Advantage $\Omega(2^{-k/2})$}
Consider any $X$ with $\delta$-smooth min-entropy smaller than $k$. 
This requirement can be seen as a statement about the ``shape'' of the distribution. 
Namely, the mass of $X$ that is above the threshold $2^{-k}$ equals at least $\delta$, that is
\begin{align*}
 \sum_{x} \max(P_X(x) - 2^{-k}, 0) \geqslant \delta.
\end{align*}
For an illustration see \Cref{fig:mass}.
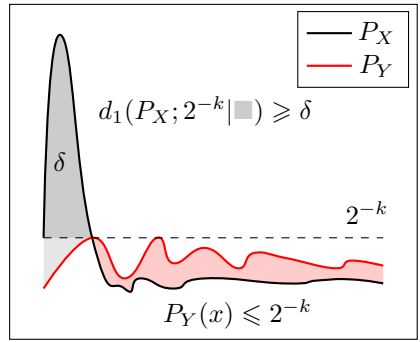
\begin{figure}[h!]
\centering
\begin{tikzpicture}
\begin{axis}[
 xtick=\empty, ytick=\empty, ymin=-0.2, legend entries = {$P_X$,$P_Y$}]
  \addplot[smooth, tension=1, thick, name path = X, color = black] coordinates { 
  (0,0) (0.1,0.4) (0.3,0) (0.5,-0.1) (0.6,-0.08) (0.8,-0.1) (1,-0.08) (1.5,-0.09) (1.7,-0.084) (1.9,-0.082) (2.1,-0.09)};
  \addplot[smooth, tension=0.8, thick, name path = Y, color = red] coordinates { (0,-0.1) (0.3,0) (0.5,-0.08) (0.7,0) (0.8,-0.05) (1,-0.02) (1.2,-0.06) (1.3,-0.03) (1.5,-0.04) (1.8,-0.06) (1.9,-0.045) (2.1,-0.055)};
  \addplot[dashed, name path = entropy] coordinates {(0,0) (2.1,0)};
  \addplot fill between[ 
    of = X and Y, 
    split, 
    every even segment/.style = {gray!20},
    every odd segment/.style  = {red!20}
  ];
  \addplot fill between[ 
    of = entropy and X, 
    split, 
    every odd segment/.style = {gray!40},
    every even segment/.style  = {yellow!40, opacity = 0}
  ];
  \node at (axis cs:1.0, 0.25) {$d_1(P_X;2^{-k}|\textcolor{gray!40}{\blacksquare}) \geqslant \delta $};
  \node at (axis cs:1.2,-0.15) {$P_Y(x) \leqslant 2^{-k}$}; 
  \node at (axis cs:0.1,0.15) {$\delta$};
  \node at (axis cs:2.0,0.05) {$2^{-k}$};
\end{axis}
\end{tikzpicture}
\caption{An intuition behind the attack. Random $\pm 1$-weights make the bias equal to the $\ell_2$-distance
of $P_X$ and $P_Y$. This distance can be bounded in terms of the $\ell_1$ distance, which concentrates mass difference $\delta$ on less than $2^k$ elements (the region in gray). }
\label{fig:mass}
\end{figure}

\noindent We construct our attack based on this observation. Define the advantage of a function $\D$ for distributions $X$ and $Y$ as
$$\mathsf{Adv}^{\D}(X;Y) = \left| \sum_{x}\D(x)(P_X(x)-P_Y(x)) \right|$$
(writing also $\mathsf{Adv}^{\D}_S $ when the summation is restricted to a subset $S$). Consider a 
random distinguisher $\D:\bin^n\rightarrow \{-1,1\}$.
Random variables $\D(x)$ for different $x$ are independent, have zero-mean and second moment equal to 1. Therefore 
the expected square of of the advantage, over the choice of $\D$, equals
\begin{align*}
 \E\left[\left(\mathsf{Adv}^{\D}(X;Y)\right)^2\right] =
 \E \left| \sum_{x}\D(x)(P_X(x)-P_Y(x)) \right|^2=
  \sum_{x}(P_X(x)-P_Y(x))^2 
\end{align*}
Let $S$ be the set of $x$ such that $P_X(x) > 2^{-k}$. For any $Y$ of min-entropy at least $k$ we obtain
\begin{align*}
 \sum_{x\in S}(P_X(x)-P_Y(x))^2  \geqslant
  \sum_{x\in S}(P_X(x)-2^{-k})^2 \geqslant |S|^{-1}\left( \sum_{x\in S}\left(P_X(x)-2^{-k}\right)\right)^2 
  \geqslant  2^{-k}\delta^2
\end{align*}
where the first inequality follows because $P_Y(x)\leqslant 2^{-k} < P_X(x)$ for $x\in S$,
 the second inequality is by the standard inequality between the first and second norm, and the third inequality follows because we showed that 
$\Pr[X\in S] \geqslant |S|\cdot 2^{-k}+\delta$ (illustrated in \Cref{fig:mass}) which also implies $|S|^{-1}\geqslant 2^{-k}$.
 
By the previous formula on the expected squared advantage this means that  
\begin{align*}
 \E\left[\left(\mathsf{Adv}^{\D}(X;Y)\right)^2\right] \geqslant  2^{-k}\delta^2
\end{align*}
for at least one choice of $\D$. This implies
\begin{align*}
\mathsf{Adv}_{}^{\D} (X;Y)\geqslant    2^{-\frac{k}{2}}\delta.
\end{align*}

A random $\D$ as defined would be of size exponential in $n$, but 
since we used only the second moment in calculations,
it suffices to generate $\D(x)$ as pairwise independent random variables. By assuming $4$-wise independence -- which can be computed by $O(n^2)$ size circuits -- we can prove slightly more, namely that a constant fraction of generated $\D$'s are good distinguishers.
This property will be important for the next step, where we amplify the advantage assuming larger distinguishers.

\paragraph{Leveraging the advantage by slicing the domain}

Consider a random and equitable partition $\{S_i\}_{i=1}^{T}$ of the set  $\{0,1\}^n$. 
From the previous analysis we know that a random distinguisher achieves advantage $\epsilon = d_2(P_X;P_Y)$ over the whole domain. 
Note that (for any, not necessarily random partition $\{S_i\}_i$) we have
$$\left(d_2(P_X;P_Y)\right)^2 = \sum_{i=1}^{T} \left(d_2(P_X;P_Y|S_i)\right)^2$$ 
where $ d_2(P_X;P_Y|S_i)$ is the restriction of the distance to the set $S_i$ (by restricting the summation to $S_i$). 
From a random partition we expect the mass difference between $P_X$ and $P_Y$ to be \emph{distributed
evenly} among the partition slices (see \Cref{fig:good_partition}). Based on the last equation, we expect
\begin{align*}
 d_2(P_X;P_Y|S_i) & \approx \frac{ d_2(P_X;P_Y)}{\sqrt{T}} 
\end{align*}
to hold with high probability over $\{S_i\}_i$. 
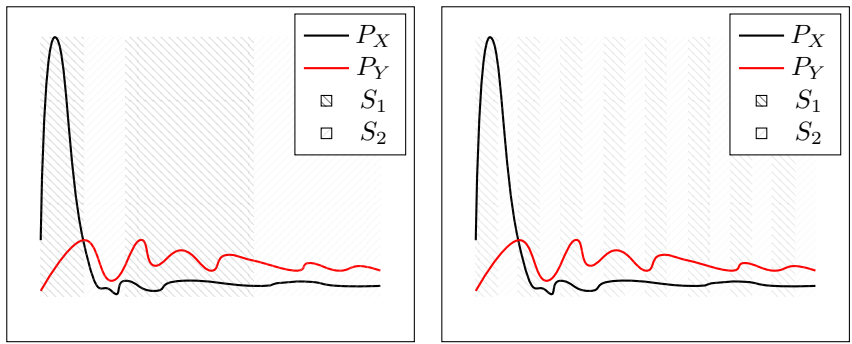
\begin{figure}[b]
\subcaptionbox{An example of a ``bad'' partition. \\ Almost all advantage is captured by one partition slice $S_1$.\label{fig:bad_partiton}}{
\begin{tikzpicture}
\begin{axis}[
 xtick=\empty, ytick=\empty, ymin=-0.2, legend entries = {$P_X$,$P_Y$,$S_1$,$S_2$}]
  \addlegendimage{thick, color = black}
  \addlegendimage{thick, color = red}
  \addlegendimage{pattern color = gray!100,only marks, mark=square*, pattern=north west lines, opacity=1}
  \addlegendimage{pattern color = gray!20,only marks, mark=square*, pattern=north east lines, opacity=1}

  \addplot[smooth, tension=1, thick, name path = X, color = black] coordinates { 
  (0,0) (0.1,0.4) (0.3,0) (0.5,-0.1) (0.6,-0.08) (0.8,-0.1) (1,-0.08) (1.5,-0.09) (1.7,-0.084) (1.9,-0.082) (2.1,-0.09) (2.4,-0.09)};
  \addplot[smooth, tension=0.8, thick, name path = Y, color = red] coordinates { (0,-0.1) (0.3,0) (0.5,-0.08) (0.7,0) (0.8,-0.05) (1,-0.02) (1.2,-0.06) (1.3,-0.03) (1.5,-0.04) (1.8,-0.06) (1.9,-0.045) (2.1,-0.06) (2.25,-0.051) (2.4,-0.06) } ;
  \draw[draw = none, pattern=north west lines, pattern color=gray!100, opacity = 0.5] (axis cs:0,-0.11) rectangle (axis cs:0.3,0.4);
  \draw[draw = none, pattern=north east lines, pattern color=gray!20, opacity = 0.5] (axis cs:0.3,-0.11) rectangle (axis cs:0.6,0.4);
  \draw[draw = none, pattern=north west lines, pattern color=gray!100, opacity = 0.5] (axis cs:0.6,-0.11) rectangle (axis cs:1.5,0.4);
  \draw[draw = none, pattern=north east lines, pattern color=gray!20, opacity = 0.5] (axis cs:1.5,-0.11) rectangle (axis cs:2.4,0.4);
\end{axis}
\end{tikzpicture}
}
\subcaptionbox{An example of a ``good`` partition. The advantage is evenly distributed among slices $S_1,S_2$.\label{fig:good_partition} }{
\begin{tikzpicture}
\begin{axis}[
 xtick=\empty, ytick=\empty, ymin=-0.2, legend entries = {$P_X$,$P_Y$,$S_1$,$S_2$}]
  \addlegendimage{thick, color = black}
  \addlegendimage{thick, color = red}
  \addlegendimage{pattern color = gray!100,only marks, mark=square*, pattern=north west lines, opacity=1}
  \addlegendimage{pattern color = gray!20,only marks, mark=square*, pattern=north east lines, opacity=1}

  \addplot[smooth, tension=1, thick, name path = X, color = black] coordinates { 
  (0,0) (0.1,0.4) (0.3,0) (0.5,-0.1) (0.6,-0.08) (0.8,-0.1) (1,-0.08) (1.5,-0.09) (1.7,-0.084) (1.9,-0.082) (2.1,-0.09) (2.4,-0.09)};
  \addplot[smooth, tension=0.8, thick, name path = Y, color = red] coordinates { (0,-0.1) (0.3,0) (0.5,-0.08) (0.7,0) (0.8,-0.05) (1,-0.02) (1.2,-0.06) (1.3,-0.03) (1.5,-0.04) (1.8,-0.06) (1.9,-0.045) (2.1,-0.06) (2.25,-0.051) (2.4,-0.06) } ;
  \foreach \i in {0.0,0.3,...,2.4}{
  \edef\temp{
  \noexpand\draw [pattern=north west lines, pattern color = gray!100, draw = none, opacity = 0.4] (axis cs:\i,-0.11) rectangle (axis cs:\i+0.15,0.4);
  \noexpand\draw [pattern=north east lines, pattern color = gray!20, draw = none, opacity = 0.4] (axis cs:\i+0.15,-0.11) rectangle (axis cs:\i+0.3,0.4);
  }
  \temp
}
\end{axis}
\end{tikzpicture}
}
\caption{Illustration of good and bad partitions.}
\label{fig:attack}
\end{figure}
In fact, if the mass difference is not well balanced amongst the slices (in the extreme case, concentrated on one slice) our argument will not offer any gain over the previous construction (see \Cref{fig:bad_partiton}).

By applying the previous argument to individual slices, for every $i$ we can obtain an advantage $\mathsf{Adv}^{\D}_{S_i}(X;Y) = \Omega\left( (T^{-\frac{1}{2}}2^{-\frac{k}{2}})\delta\right)$ when restricted to the set $S_i$ (with high probability over the choice of $\D$ and $\{S_{i}\}_i$). Now if the sets $S_i$ are \emph{efficiently recognizable},
we can combine them into a better distinguisher. Namely for every $i$ we chose a value $\beta_i\in\{-1,1\}$ such that 
$\D$'s advantage (before taking the absolute value) restricted to $S_i$ has sign $\beta_i$,  and set 
\begin{align*}
\hat{\D}(x) =\beta_i \D(x), \text{ where $i$ is such that } x\in S_i,
\end{align*}
then the advantage equals (with high probability over  $\D$ and the $S_i$'s)
\begin{align*}
 \mathsf{Adv}^{\hat{\D}} (X;Y)= \sum_{i=1}^{T}\mathsf{Adv}^{\D}_{S_i}(X;Y) = \Omega\left( T^{\frac{1}{2}}2^{-\frac{k}{2}} \delta \right)
\end{align*}
We need to specify a $4$-wise independent hash for $\D$, another $4$-wise independent hash for deciding in which of the $T$ slices an element lies, and $T$ bits to encode the $\beta_i$'s. Thus for a given $T$ the size of $\hat{\D}$ will be $T+\tilde O(n)$. Using the above equation, we then get a smooth tradeoff $s =O(2^{k}\epsilon^{2}\delta^{-2})$ between the advantage $\epsilon$ and the circuit size $s$.
This discussion shows that to complete the argument we need the following two properties of the partition (a) the mass difference between $P_X$ and $P_Y$ is (roughly) equidistributed among slices
and (b) the membership in partition slices can be efficiently decided.
\paragraph{Slicing using $4$-wise independence}
To complete the argument, we assume that $T$ is a power of $2$, and generate the slicing by using a $4$-universal hash function $h:\bin^n\rightarrow \bin^{\log{T}}$. The $i$-th slice $S_i$ is defined as $\{x\in\bin^n: h(x)=i\}$. 
These assumptions are enough to prove that
\begin{align*}
 \E\mathsf{Adv}^{\hat\D}_{S_i}(X;Y) =\Omega\left( T^{-\frac{1}{2}}d_2(P_X;P_Y) \right) =   \Omega\left( T^{-\frac{1}{2}}2^{-\frac{k}{2}}\delta \right).
\end{align*}
Interestingly, the expected advantage (left-hand side) cannot be computed directly. The trick here is to bound it in terms of the second and fourth moment.
The above inequality, coupled with bounds on second moments of the advantage $\mathsf{Adv}^{\hat\D}_{S_i}$ (obtained directly), allows us to prove that
\begin{align*}
\Pr\left[  \sum_{i=1}^{T}\mathsf{Adv}^{\hat\D}_{S_i} \geqslant \Omega(1)\cdot  T^{\frac{1}{2}}2^{-\frac{k}{2}} \delta \right] > \Omega(1).
\end{align*}
This shows that there exists the claimed distinguisher $\hat{\D}$. In fact,
a \emph{constant fraction} of generated (over the choice of $\D$ and $\{\S_i\}_i$) distinguishers $\hat{\D}$'s works.  


\paragraph{Random walks}
From a technical point of view, our method involves computing higher moments of the advantages to obtain concentration and anti-concentration results. 
The key observation is that the advantage written down as 
\begin{align*}
\mathsf{Adv}^{\D}_{S_i}(X;Y) = \left|\sum_{x}(P_X(x)-P_Y(x))\mathbf{1}_{S_i}(x)\D(x)\right|
\end{align*}
which can be then studied as a \emph{random walk}
\begin{align*}
\mathsf{Adv}^{\D}_{S_i}(X;Y) = \left|\sum_{x}\xi_{i,x}\right|
\end{align*}
with zero-mean increments $\xi_{i,x} = (P_X(x)-P_Y(x))\mathbf{1}_{S_i}(x)\D(x)$.
The difference with respect to classical model is that the increments are only
$\ell$-wise independent (for $\ell=4$). However, that classical moment bounds still apply (see \Cref{sec:random_walks,sec:anticon_bounds} for more details).




\section{Preliminaries}

\subsection{Interpolation Inequalities}\label{sec:interpol}

Interpolation inequalities show how to bound the $p$-th moment of a random variable if we know bounds on one smaller and one higher moment.
The following result is known also as \emph{log-convexity of $L_p$ norms}, and can be proved by the H\"{o}lder Inequality

\begin{lemma}[Moments interpolation]\label{lemma:moments_interpolation}
For any $p_1 < p < p_2$ and any bounded random variable $Z$ we have
\begin{align*}
  \|Z\|_p \leqslant \left(\|Z\|_{p_1}\right)^{\theta} \left(\|Z\|_{p_2}\right)^{1-\theta}
\end{align*}
where $\theta$ is such that $\frac{\theta}{p_1} + \frac{1-\theta}{p_2} = \frac{1}{p}$, and
for any $r$ we define $\|Z\|_r = \left(\E|Z|^r\right)^{\frac{1}{r}}$.
\end{lemma}

Alternatively, we can \emph{lower bound} a moment given \emph{two higher moments}. 
This is very useful when higher moments are easier to compute. In this work will bound
first moments from below when we know the second and the fourth moment (which are easier to compute as they are even-order moments)
\begin{corollary}\label{cor:interpolation}
For any bounded $Z$ we have $\E|Z| \geqslant \frac{\left(\E|Z|^2\right)^{\frac{3}{2}}}{\left(\E|Z|^4\right)^{\frac{1}{2}}} $.
\end{corollary}

\subsection{Moments of random walks}\label{sec:random_walks}

For a random walk $\sum_{x}\xi(x)$, where $\xi(x)$ are independent with zero-mean, we have good control over the moments, namely
$\E\left| \sum_{x}\xi(x)\right|^p = \Theta(1)\cdot \left(\sum_{x}\mathrm{Var}(\xi(x))\right)^{\frac{p}{2}}$ where constants depend on $p$.
This result is due to Marcinkiewicz and Zygmund \cite{marcinkiewicz1937quelques} who extended the former result of Khintchine 
\cite{khintchine1924satz}.
Below we notice that for \emph{small moments} $p$ it suffices to assume only $p$-wise independence (most often used versions 
assume fully independence)
\begin{lemma}[Strengthening of Marcinkiewicz-Zygmund's Inequality for $p=4$]\label{lemma:Khintchine}
Suppose that $\{\xi(x)\}_{x\in\mathcal{X}}$ are $4$-wise independent, with zero mean. Then we have
\begin{align*}
\frac{1}{\sqrt{3}} \left(\sum_{x\in\mathcal{X}}\mathrm{Var}(\xi(x))\right)^{\frac{1}{2}} \leqslant  & \E \left|\sum_{x\in\mathcal{X}} \xi(x)\right| \leqslant   \left(\sum_{x\in\mathcal{X}}\mathrm{Var}(\xi(x))\right)^{\frac{1}{2}}  \\
 & \E \left|\sum_{x\in\mathcal{X}} \xi(x)\right|^2   =   \sum_{x\in\mathcal{X}}\mathrm{Var}(\xi(x)) \\
\left(\sum_{x\in\mathcal{X}}\mathrm{Var}(\xi(x))\right)^2 \leqslant & \E \left|\sum_{x\in\mathcal{X}} \xi(x)\right|^4 \leqslant  3 \left(\sum_{x\in\mathcal{X}}\mathrm{Var}(\xi(x))\right)^2 
\end{align*}
\end{lemma}
The proof appears in \Cref{proof:lemma:Khintchine}.

\subsection{Anticontentration bounds}\label{sec:anticon_bounds}

\begin{lemma}[Paley-Zygmund Inequality]\label{lemma:PaleyZygmund}
For any positive random variable $Z$ and a parameter $\theta\in(0,1)$ we have
\begin{align*}
 \Pr\left[Z > \theta\E Z \right] \geqslant (1-\theta)^2\frac{(\E Z)^2}{\E Z^2}.
\end{align*}
\end{lemma}
By applying \Cref{lemma:PaleyZygmund} to the setting of \Cref{lemma:Khintchine}, and choosing $\theta = \frac{1}{\sqrt{3}}$ we obtain
\begin{corollary}[Anticoncentration for walks with $4$-wise independent increments]\label{cor:anticonc_4independent_walks}
Suppose that $\{\xi(x)\}_{x\in\mathcal{X}}$ are $4$-wise independent with zero-mean, then we have
\begin{align*}
\Pr\left[ \left| \sum_{}\xi(x) \right| >  \frac{1}{3}\left(\sum_{}\mathrm{Var}(\xi(x))\right)^{\frac{1}{2}} \right] > \frac{1}{17}.
\end{align*}
where the summation is over $x\in\mathcal{X}$.
\end{corollary}

\section{Proof of \Cref{thm:main}}

\begin{lemma}[Characterizing smooth min-entropy]\label{lemma:smooth_charact}
For any random variable $X$ with values in a finite set $\mathcal{X}$, any $\delta$ and $k$ we have the following equivalence
$$ H_{\infty}^{\delta}(X) \geqslant k \iff \sum_{x\in\mathcal{X}}\max\left(P_X(x)-2^{-k},0\right) \leqslant \delta.$$ 
\end{lemma}
The proof appears in \Cref{proof:lemma:smooth_charact}. We will work with the following equivalent statement
\begin{corollary}[No smooth min-entropy $k$ implies bias w.r.t. distributions of min-entropy $k$ over at most $2^k$ elements]
\label{cor:no_smoothentropy_implies_bias_to_entropy}
We have $\Hmins{\delta}(X) < k$ if and only if there exists a set $S$ of at most $2^k$ elements such that
$$ \sum_{x\in S}\left|P_X(x)-P_Y(x)\right| > \delta $$
for all $Y$ of min-entropy at least $k$.
\end{corollary}

\begin{proof}[Proof of \Cref{cor:no_smoothentropy_implies_bias_to_entropy}]
The direction $\Longleftarrow$ trivially follows by the definition of smooth min-entropy. 
Now assume $\Hmins{\delta}(X) < k$. Let $S$ be the set of all $x$ such that $P_X(x) > 2^{-k}$, then 
$|S|<2^k$, and moreover by \Cref{lemma:smooth_charact} we have $\sum_{x\in S}\left(P_X(x)-2^{-k}\right)>\delta$. In particular for any $Y$ of min-entropy $k$ (i.e., $P_Y(x)\leqslant 2^{-k}$ for all $x$) 
$$ \sum_{x\in S}\left(P_X(x)-P_Y(x)\right) > \delta  $$
\end{proof}

\begin{lemma}[Bias implies Euclidean distance]\label{lemma:bias_to_euclidean}
For any distributions $P_X,P_Y$ on $\mathcal{X}$ and any subset $S$ of $\mathcal{X}$ we have
$$ \left(\sum_{x\in S}\left( P_X(x)-P_Y(x) \right)^2\right)^{\frac{1}{2}} > |S|^{-1/2}\sum_{x\in S}\left| P_X(x)-P_Y(x) \right|. $$
\end{lemma}
\begin{proof}
By the Jensen Inequality we have 
\begin{align*}
|S|^{-1}\left(\sum_{x\in S}\left( P_X(x)-P_Y(x) \right)^2\right) > \left(|S|^{-1}\sum_{x\in S}\left| P_X(x)-P_Y(x)\right|\right)^2 
\end{align*}
which is equivalent to the statement.
\end{proof}

\begin{corollary}[No smooth min-entropy implies Euclidean distance to min-entropy distributions]\label{cor:no_smooth_implies_euclidean_distance}
Suppose that $\Hmins{\delta}(X) < k$. Then for any $Y$ of min-entropy at least $k$ we have
$\left(\sum_{x}|P_X(x)-P_Y(x)|^2\right)^{\frac{1}{2}} > 2^{-\frac{k}{2}}\delta$.
\end{corollary}
\begin{proof}[Proof of \Cref{cor:no_smooth_implies_euclidean_distance}]
It suffices to combine \Cref{lemma:bias_to_euclidean} and \Cref{cor:no_smoothentropy_implies_bias_to_entropy}.

\end{proof}
By \Cref{cor:anticonc_4independent_walks} we conclude that the advantage of a random distinguisher for any two measures (in our case $P_X$ and $P_Y$) equals the Euclidean distance.
\begin{lemma}[The advantage of a random distinguisher equals the Euclidean distance]\label{lemma:random_attack_advantage}
Let $\{\D(x)\}_{x\in\{0,1\}^n}$ be $4$-wise independent as indexed by $x$ and such that $\D(x)$ outputs a random element from $\{-1,1\}$. 
Then for any 
set $S$ we have
\begin{align*}
 \left| \sum_{x\in S}\D(x)(P_X(x)-P_Y(x)) \right|  > \frac{1}{3}\cdot d_2(P_X;P_Y)
\end{align*}
with probability $\frac{1}{17}$ over the choice of $\D$ (the result actually holds for any measures in place of $P_X,P_Y$).
\end{lemma}
For our case, that is the setting in \Cref{lemma:bias_to_euclidean}, we obtain
\begin{corollary}[A random attack achieves $\Omega\left(2^{-k}\delta\right)$ with significant probability]\label{cor:random_attack_advantage}
For $X,Y$ as in \Cref{cor:no_smooth_implies_euclidean_distance}, and $\D$ as in \Cref{lemma:random_attack_advantage} we have 
 $\mathsf{Adv}^{\D}(X;Y) \ge \frac{1}{3}\cdot 2^{-\frac{k}{2}}\delta$ w.p.  $\frac{1}{17}$ over $\D$.
\end{corollary}

\subsection{Partitioning the domain into $T$ slices}

Let $h:\{0,1\}^n\rightarrow [1\ldots 2^t]$, where $t=\lceil \log T\rceil$,
be a $4$-universal hash function. Define $S_{i} = \{x: h(x) = i\}$,
$\Delta(x) = P_X(x)-P_Y(x)$ 
and consider advantages on slices $S_i$
\begin{align*}
 \mathsf{Adv}^{\D}_{S_i}\left(X;Y\right) & = \left| \sum_{x}\Delta(x)\D(x)\mathbf{1}_{S_i}(x) \right|
\end{align*}

The following corollary shows that on each of our $T$ slices, we get the advantage $T^{-\frac{1}{2}}2^{-\frac{k}{2}}\delta$.
The proof appears in \Cref{proof:cor:slice_advantages}.
\begin{corollary}[(Mixed) moments of slice advantages]\label{cor:slice_advantages}
For $\D$, $\{S_u\}_u$ as above and every $i,j$
\begin{align*}
 \E\limits_{\D,\{S_u\}_u}\mathsf{Adv}^{\D}_{S_i}(X;Y) & \geqslant 3^{-\frac{1}{2}} T^{-\frac{1}{2}} \cdot d_2\left(P_X;P_Y\right)  \\
 \E\limits_{\D,\{S_u\}_u}\left( \mathsf{Adv}^{\D}_{S_i}\left(X;Y\right)\mathsf{Adv}^{\D}_{S_j}(X;Y)\right)
 & \leqslant T^{-1} \cdot d_2(P_X;P_Y)^2
\end{align*}
(the statement is valid for arbitrary measures in place of $P_X,P_Y$).
\end{corollary}
Denote $Z = \sum_{i} \mathsf{Adv}^{\D}_{S_i}\left(X;Y\right)$. Using \Cref{lemma:PaleyZygmund} with $\theta = \frac{1}{\sqrt{3}}$ where 
we compute $\E Z^2$ and $\E Z$ according to \Cref{cor:slice_advantages} we obtain
$\Pr\left[ |Z| > \frac{1}{\sqrt{3}}\cdot \E|Z|\right] \geqslant \frac{1}{17} $. Bounding
once again $\E|Z|$ as in  \Cref{cor:slice_advantages} we get
\begin{corollary}[Total advantage on all parition slices]\label{cor:parition_advantage}
For $X,Y$ as in \Cref{cor:no_smooth_implies_euclidean_distance}, $\D$ and $S_i$ defined above we have
\begin{align*}
 \Pr_{\D,\{S_u\}_u}\left[ \sum_{i=1}^{T} \mathsf{Adv}^{\D}_{S_i}(X;Y) \geqslant \frac{1}{3}\cdot T^{\frac{1}{2}}2^{-\frac{k}{2}}\delta \right] \geqslant 
\frac{1}{17}.
\end{align*}
(for general $X,Y$ the lower bound is $\Omega(1)\cdot T^{\frac{1}{2}}\cdot d_2(P_X;P_Y)$).
\end{corollary}
The corollary shows that the \emph{total absolute advantage} over all partition slices, is as expected. Since $\{S_i\}_i$ is a partition we have
\begin{align*}
 \sum_{i=1}^{T} \mathsf{Adv}^{\D}_{S_i}(X;Y) & = \sum_{i=1}^{T}\left|\sum_{x\in S_i}\left(P_X(x)-P_Y(x)\right)\D(x) \right|  = \sum_{x}\left(P_X(x)-P_Y(x)\right)\D(x)\beta(x)
\end{align*}
where for $\beta_i\overset{def}{=}\mathrm{sgn}\left( \sum_{x\in S_i}\left(P_X(x')-P_Y(x)\right)\D(x) \right)$ (the sign of
the advantage on the $i$-th slice) we define
$\beta(x) = \beta_i$ where $S_i$ contains $x$. This shows that by ''flipping`` the distinguisher output on the slices
we achieve the sum of individual advantages. Since the bit $\beta(x)$ can be computed with $O(T) + \tilde{O}(n)$ advice (the complexity of the function $i\rightarrow \beta_i$ plus the complexity of finding $i$ for a given $x$)
we obtain
\begin{corollary}[Computing total advantage by one distinguisher]\label{cor:parition_advantage_final}
For $X,Y$ as in \Cref{cor:no_smooth_implies_euclidean_distance}, $\D$ and $\{S_i\}_i$ defined above 
there exists a modification to $\D$ which in time $\tilde{O}(n)$ and advice $O(T)$ achieves advantage $\frac{1}{3}\cdot T^{\frac{1}{2}}2^{-\frac{k}{2}}\delta$ with probability
$\frac{1}{17}$.
\end{corollary}

Finally by setting $\epsilon =  T^{\frac{1}{2}}2^{-\frac{k}{2}}\delta$ and manipulating $T$ we arrive at
\begin{corollary}[Continue tradeoff]\label{cor:tradeoof}
For any $\epsilon$ there exists $T$ such that the distinguisher in \Cref{cor:parition_advantage_final}
has advantage $\epsilon$ and circuit complexity $s = O\left(2^{k}\epsilon^2\delta^{-2}\right)$.
\end{corollary}

\section{Omitted Proofs}


\subsection{Proof of \Cref{lemma:Khintchine} (Strengthening of Marcinkiewicz-Zygmund's Inequality for $p=4$)}\label{proof:lemma:Khintchine}

Let $Z = \sum_{x}\xi(x)$. Since $\xi(x)$ are (in particular) $2$-wise independent with zero mean, we get
\begin{align*}
\E\left(\sum_{x} \xi(x)\right)^2  = \sum_{x,y}\E\left( \xi(x)\xi(y) \right) 
= \sum_{x=y}\E\left( \xi(x)\xi(y) \right) = \sum_{x}\mathrm{Var}(\xi(x)).
\end{align*}
(the summation taken over $x,y\in\mathcal{X}$). The fourth moment is somewhat more complicated
\begin{align*}
\E\left(\sum_{x} \xi(x)\right)^4 & = \sum_{x_1,x_2,x_3,x_4}\E\left( \xi(x_1)\xi(x_2)\xi(x_3)\xi(x_4) \right) \\
& = \sum_{x_1=x_2=x_3=x_4}\E\left( \xi(x_1)\xi(x_2)\xi(x_3)\xi(x_4)\right) + \\
& \quad + 3\sum_{x_1=x_2\not=x_3=x_4}\E\left( \xi(x_1)\xi(x_2)\xi(x_3)\xi(x_4)\right) \\
& = \sum_{x}\E\xi(x)^4 + 3\sum_{x\not=y}\E\xi(x)^2\E\xi(y)^2 \\
& = 3\left(\sum_{x}\E \xi(x)^2\right)^2-2\sum_{x}\E \xi(x)^4
\end{align*}
The second equality follows because whenever $\xi(x)$ occurs in an odd power, for example $x=x_1\not=x_2=x_3=x_4$, the expectation is zero (this way one can
 simplify and bound also higher moments, see \cite{DBLP:conf/soda/SchmidtSS93}).
It remains to estimate the first moment. By \Cref{cor:interpolation} and bounds on the second and fourth moment we have just computed we obtain
\begin{align*}
 \frac{1}{\sqrt{3}}\cdot  \left(\sum_{x\in\mathcal{X}}\mathrm{Var}(\xi(x))\right)^{\frac{1}{2}} 
 \leqslant \E \left|\sum_{x\in\mathcal{X}} \xi(x)\right| 
\end{align*}
and the upper bound follows by Jensen's Inequality (with constant 1).


\subsection{Proof of \Cref{lemma:smooth_charact} (Characterizing smooth min-entropy)}
\label{proof:lemma:smooth_charact}
Suppose that $\Hmins{\delta}(X) \geqslant k$.
then, by definition, there is $Y$ such that $\Hmin(Y) \geqslant k$ and
$\sum_{x: P_X(x)>P_Y(x)} P_X(x)-P_Y(x) \leqslant \delta$. Since all the summands are positive and since $P_Y(x) \leqslant 2^{-k}$, ignoring those $x$ for which
$P_Y(x) < 2^{-k}$ yields
\begin{align*}
 \sum_{x: P_X(x)>2^{-k}} P_X(x)-P_Y(x) \leqslant \delta.
\end{align*}
Again, since $P_Y(x) \leqslant 2^{-k}$ we obtain
\begin{align*}
 \sum_{x: P_X(x)>2^{-k}} P_X(x)-2^{-k} \leqslant \delta,
\end{align*}
which finishes the proof of the ''$\Longrightarrow$`` part. 

Assume now that $\delta' = \sum_{x\in\mathcal{X}}\max\left(P_X(x)-2^{-k},0\right) \leqslant \delta$.
Note that 
\begin{align*}
 \sum_{x\in\mathcal{X}}\max\left(P_X(x)-\frac{1}{2^{k}},0\right) & + \sum_{x\in\mathcal{X}}\max\left(\frac{1}{2^{k}}-P_X(x),0\right)  = \\
 & = 
2\sum_{x\in\mathcal{X}}\left|P_X(x)-\frac{1}{2^{k}}\right| 
 \geqslant 2 \sum_{x\in\mathcal{X}}\max\left(P_X(x)-\frac{1}{2^{k}},0\right)
\end{align*}
and therefore we have $\sum_{x\in\mathcal{X}}\max\left(2^{-k}-P_X(x),0\right) \geqslant \delta'$. 
By this observation we can construct a distribution $Y$ by shifting $\delta'$ of the mass of $P_X$ from the set $S^{-}=\{x:P_X(x)>2^{-k}\}$
to the set  $\{x:2^{-k} \geqslant P_X(x)\}$ in such a way that we have $P_{Y}(x) \leqslant 2^{-k}$ for all $x$.
Thus $\Hmin(Y) \geqslant k$ and since a $\delta'$ fraction of the mass is shifted and redistributed we have $d_1(X;Y) \leqslant \delta'$.
This finishes the proof of the ''$\Longleftarrow$`` part.


\subsection{Proof of \Cref{cor:slice_advantages} ((Mixed) moments of slice advantages)}\label{proof:cor:slice_advantages}

For shortness denote $\Delta(x) = P_X(x)-P_Y(x)$ and
$ \mathsf{Adv}^{\D}_{S_i}=\mathsf{Adv}^{\D}_{S_i}\left(X;Y\right)$.

Note that by \Cref{lemma:Khintchine}, applied to
the family $f_x = \Delta(x)\D(x)\mathbf{1}_{S_i}(x)$ (which is $4$-wise independent) we have
\begin{align*}
   \E \mathsf{Adv}^{\D}_{S_i} \geqslant 3^{-\frac{1}{2}}\left(\sum_{x}\Delta(x)^2\right)^{\frac{1}{2}}
\end{align*}
which is the first inequality claimed in the corollary.\Note{how does $T$ come in here?}
In turn, again by \Cref{lemma:Khintchine}, we have
\begin{align*}
 \E\left( \mathsf{Adv}^{\D}_{S_i}\right)^2 = T^{-1} \cdot \sum_{x}\Delta(x)^2.
\end{align*}
Since this holds for any $i$, by Cauchy-Schwarz we get for any $i,j$
\begin{align*}
 \E\mathsf{Adv}^{\D}_{S_i}  \mathsf{Adv}^{\D}_{S_j} \leqslant \sqrt{\E\left( \mathsf{Adv}^{\D}_{S_i}\right)^2\cdot \E\left( \mathsf{Adv}^{\D}_{S_j}\right)^2 } \leqslant T^{-1} \cdot \sum_{x}\Delta(x)^2.
\end{align*}
which proves the second inequality in the corollary.

\printbibliography

\end{document}